\documentclass[copyright,creativecommons]{eptcs}
\providecommand{\event}{QAPL 2015} % Name of the event you are submitting to
\usepackage{breakurl}             % Not needed if you use pdflatex only.
\usepackage[caption=false,font=footnotesize,captionskip=14pt]{subfig} % "caption=false,font=footnotesize" required by IEEE
\usepackage{latexsym}
\usepackage[cmex10]{amsmath} % "cmex10" required by IEEE
\usepackage{amssymb}
\usepackage{times}
\usepackage{stmaryrd}
\usepackage{amsfonts}
\usepackage{mathrsfs}
\usepackage{color}
\newcommand{\withproof}[1]{}

\renewcommand{\emph}[1]{{\it #1}}

\newif\ifcommentson\commentsonfalse
\ifcommentson
\newcommand{\commentTGW}[1]{\begin{center} \parbox{.9\textwidth}{\textbf{\textcolor{black}{Comment T.}} \textcolor{red}{#1 }}\end{center}}
\newcommand{\commentYK}[1]{\begin{center} \parbox{.9\textwidth}{\textbf{\textcolor{black}{Comment Y.}} \textcolor{red}{#1} }\end{center}}
\newcommand{\replyTGW}[1]{\begin{center} \parbox{.8\textwidth}{\textbf{Reply T.} \textcolor{blue}{#1} }\end{center}}
\newcommand{\replyYK}[1]{\begin{center} \parbox{.8\textwidth}{\textbf{Reply Y.} \textcolor{blue}{#1} }\end{center}}
\newcommand{\commentT}[1]{\marginpar{\footnotesize \color{red} {\bf T:} \textsf{\scriptsize #1}}}
\newcommand{\commentY}[1]{\marginpar{\footnotesize \color{red} {\bf Y:} \textsf{\scriptsize #1}}}
\newcommand{\replyT}[1]{\marginpar{\footnotesize \color{red} {\bf T:} \textsf{\scriptsize #1}}}
\newcommand{\replyY}[1]{\marginpar{\footnotesize \color{red} {\bf Y:} \textsf{\scriptsize #1}}}
\else
\newcommand{\commentTGW}[1]{}
\newcommand{\commentYK}[1]{}
\newcommand{\replyTGW}[1]{}
\newcommand{\replyYK}[1]{}
\newcommand{\commentT}[1]{}
\newcommand{\commentY}[1]{}
\newcommand{\replyT}[1]{}
\newcommand{\replyY}[1]{}
\fi
\newcommand{\withpossible}[1]{}

\newcommand{\hacks}[1]{#1} % Turn hacks on

\newtheorem{theorem}{Theorem}
\newtheorem{proposition}[theorem]{Proposition}

\newtheorem{corollary}[theorem]{Corollary}

\newtheorem{definition}{Definition}
\newtheorem{example}{Example}

\newenvironment{proof}{\noindent\textbf{Proof: }}{\hspace*{\fill} $\Box$\\}

\newcommand{\leakiEst}{\textsf{leakiEst}}

\newcommand{\lpsolve}{\textsf{lp\_solve}}

\newcommand{\X}{\mathcal{X}}
\newcommand{\Y}{\mathcal{Y}}
\renewcommand{\O}{\mathcal{O}}
\newcommand{\Obs}{\mathit{Obs}}
\newcommand{\Z}{\mathcal{Z}}
\newcommand{\Int}{\mathit{Int}}
\newcommand{\sep}{\mathit{sep}}

\newcommand{\oap}[2]{\overline{#1}\langle #2\rangle}
\newcommand{\define}{\stackrel {\rm def} =}
\newcommand{\chan}{\mathcal{K}}
\renewcommand{\S}{\mathcal{S}}
\newcommand{\SDS}{\S_{\mathit{DS}}}
\newcommand{\SFS}{\S_{\mathit{FS}}}
\newcommand{\SFI}{\S_{\mathit{FI}}}

\newcommand{\comps}[1]{\mathit{Comp}_{#1}}
\newcommand{\compd}{\mathbin{\times}}

\newcommand{\I}{\mathcal{I}}
\newcommand{\C}{\mathcal{SC}}
\renewcommand{\L}{\mathcal{L}}
\newcommand{\MC}{\mathcal{MC}}

\title{Quantitative Information Flow \\ for Scheduler-Dependent Systems
\thanks{This work has been partially supported by the project ANR-12-IS02-001 PACE, by the INRIA Equipe Associ\'{e}e PRINCESS, by the INRIA Large Scale Initiative CAPPRIS, and by  EU grant agreement no. 295261 (MEALS). The work of Yusuke Kawamoto has been supported by a postdoc grant funded by the IDEX Digital Society project.}}
\author{%
Yusuke Kawamoto
\institute{\mbox{ }Inria Saclay \& LIX, \'{E}cole Polytechnique, France}
\and
Thomas Given-Wilson
\institute{Inria, France}
}
\def\titlerunning{Quantitative Information Flow for Scheduler-Dependent Systems}
\def\authorrunning{Y. Kawamoto \& T. Given-Wilson}

\begin{document}
\pagestyle{plain} 

\maketitle  

\begin{abstract}
Quantitative information flow analyses measure how much information on secrets is leaked by publicly observable outputs.
One area of interest is to quantify and estimate the information leakage of composed systems.
Prior work has focused on running disjoint component systems in parallel and reasoning about the leakage compositionally, but has not explored how the component systems are run in parallel or how the leakage of composed systems can be minimised.
In this paper we consider the manner in which parallel systems can be
combined or scheduled. This considers the effects of scheduling channels where
resources may be shared, or whether the outputs may be incrementally observed.
We also generalise the attacker's capability, of observing outputs of the system, 
to consider attackers who may be imperfect in their observations,
e.g.~when outputs may be 
confused with one another, or when assessing the time taken for an output to appear.
Our main contribution is to present how scheduling and observation effect information 
leakage properties.
In particular, that scheduling can hide some leaked information from 
perfect observers, while some scheduling may reveal secret information that 
is hidden to imperfect observers.
In addition we present an algorithm to construct a scheduler that minimises 
the min-entropy leakage and min-capacity in the presence of any observer.
%
%\keywords{Quantitative information flow $\cdot$
%Information leakage $\cdot$
%Scheduling $\cdot$
%Composition $\cdot$
%Observer $\cdot$
%Min-entropy leakage}
\end{abstract}

\renewcommand{\arraystretch}{1.5}

\section{Introduction}
\label{sec:intro}

Preventing the leakage of confidential information is an important goal in research of information security.
When some information leakage is unavoidable in practice, the next step is to quantify and reduce the leakage.
Recently theories and tools on quantitative information flow have been developed using information theory
to address these issues~\cite{Clark:01:QAPL,Boreale:09:InfComput,Koepf:07:CCS,Chatzikokolakis:08:IC,Smith:09:FOSSACS,Boreale:11:FOSSACS,ChothiaKawamoto2013,ChothiaKN14:esorics}.
The common approach is to model systems as \emph{information-theoretic channels} that receive secret input and returns observable output.

One area of interest is to quantify and estimate the information leakage of composed systems.
When composing systems the manner of reasoning about their behaviour is
non-trivial and is complicated by many factors.
One of the first approaches is to consider the \emph{(disjoint) parallel composition}, that is, simply running the component systems independently and regarding them as a single composed system.
This approach provides some general behaviour and reasoning about the whole composed system, as shown in the research of quantitative information flow with different operational scenarios of attack~\cite{barthe2011information,espinoza2013min,KawamotoCP14:qest}.
However, the parallel composition approach is coarse-grained and abstracts many of the channels'
behaviours that may lead to changes in information leakage.
Although this approach provides useful results on the bounds of possible leakage,
it does so under the assumption that the component channels are executed independently and 
observed separately.
That is, their outputs can always be linked to the disjoint component channels,
and that both their outputs are observed simultaneously and without any interleaving or reflection of
how the component channels achieved their outputs.

Here we take a more fine-grained approach where we consider that channels may provide a
sequence of observable actions. Thus, a channel may be observed to output a
sequence of actions, or the passage of time may be observed to pass between the
initiation of the channel and a final output.
This captures more mechanics of real world systems and allows for greater
refined reasoning about their behaviour.

Such sequences of observable actions also allow a more subtle approach to combining 
channels in parallel. Rather than simply taking both outputs to
appear together at the termination of their operations, observations can be made of
the sequence in which the outputs appear. Such a combination of channels
becomes parametrised by a \emph{scheduler}, that informs on how to combine the observable
sequences of actions into a single sequence. This can then represent very direct
behaviour such as scheduling properties of a shared CPU, or abstract behaviours
such as routing properties, vote counting, etc.

The other novel approach presented here is the refinement of the attacker's capability of observing the outputs of systems.
We model attackers that may have imperfect observability:
they may not accurately detect differences in outputs, or may do so only probabilistically.
This captures, for example,
the situation 
where the attacker may be blind to some internal behaviour that other agents can detect.
In this paper such imperfect observations are modeled using 
what we call \emph{observer channels}.
This formalisation enables us to consider a large class of observers, including \emph{probabilistic observers}, which have never been considered in the previous studies on quantitative information flow.

These refinements to composing information-theoretic channels allow us to reason
about behaviours that may be obvious, but not captured by previous approaches.
In this paper we present three kinds of results regarding the effect of leakage properties
due to the considering of schedulers and observers.
First, since scheduled composition can alter the leakage relative to the parallel composition,
we present theorems for detecting when a scheduled composition does not alter the relative information leakage.
This means some preliminary analysis may be sufficient to determine when scheduled composition
may be worthy of further consideration.
Second, scheduled composition can leak more or less information than the parallel composition depending on
the properties of the channels and the power of the observer.
Although the potential effect on leakage is dependent upon many factors, we present results
that determine an upper bound for the leakage of a schedule-composed channel.
Third, we present results for finding a scheduler that minimises the min-entropy leakage and min-capacity 
in the presence of any observer.
We present how to construct such a scheduler by solving a linear programming problem.

In addition we evaluate our model and results with some simple yet intuitive examples,
such as mix networks for voter anonymity, and side-channel attacks against cryptographic
algorithms.
We provide an implementation that can be used to calculate the behaviours of information-theoretic channels, schedulers, and observers as presented here.
The implementation is available online~\cite{evils:www}, which requires the libraries \leakiEst{} tool~\cite{chothia13:cav} 
and the linear programming system \lpsolve{}~\cite{lpsolve}.

The rest of the paper is structured as follows.
Section~\ref{sec:preliminaries} recalls the definitions of information-theoretic channels and measures of information leakage.
Section~\ref{sec:comp-view} defines traces, systems and channel compositions, and shows examples of schedulers.
Section~\ref{sec:observed-leak} introduces the notion of generalised observers and defines the observed leakage.
Section~\ref{sec:main} presents our main results in a general manner.
Section~\ref{sec:eval} applies these results to well known problems.
Section~\ref{sec:related} discusses some related work.
Section~\ref{sec:conclude} draws conclusions and discusses future work. 
All proofs can be found in~\cite{KawamotoGivenWilson:15:HAL}.

\hacks{\vspace{-0.3cm}}
\section{Preliminaries}
\label{sec:preliminaries}
\vspace{-0.2cm}
\subsection{Information-Theoretic Channel}
\label{ssec:channel}

Systems are modeled as \emph{information-theoretic channels} to quantify information leakage using information theory.
A channel $\chan$ is defined as a triple $(\X, \Y, C)$ consisting of
a finite set $\X$ of secret input values,
a finite set $\Y$ of observable output values,
and a \emph{channel matrix} $C$ each of whose row represents a probability distribution;
i.e., for all $x \in \X$ and $y \in \Y$, $0 \le C[x, y] \le 1$ and $\sum_{y' \in \Y} C[x, y'] = 1$.
For each $x \in \X$ and $y \in \Y$, $C[x, y]$ is a conditional probability
$p(y|x)$ of observing $y$ when the secret of the system is $x$.
We assume some secret distribution $\pi$ on $\X$, which is also called a \emph{prior}.
Given a prior $\pi$ on $\X$, the joint distribution of having a secret $x \in \X$ and an observable $y \in \Y$ is defined by $p(x, y) = \pi[x] C[x, y]$.

\subsection{Quantitative Information Leakage Measures}
\label{subsec:info-leak}

In this section we recall the definitions of two popular quantitative information leakage measures.

Mutual information is a leakage measure based on the Shannon entropy of the secret  distribution.
\begin{definition} \label{def:MI} \rm
Given a prior $\pi$ on $\cal X$ and a channel $\chan = ({\cal X}, {\cal Y}, C)$,
the \emph{mutual information} $\I(\pi, \chan)$ w.r.t. $\pi$ and $\chan$ is defined by:
\vspace{-0.3cm}
\[
\I(\pi, \chan) = \sum_{x \in {\cal X}, y \in {\cal Y}}
\pi[x] C[x, y] \log\left( \frac{ C[x, y] }{ \sum_{y' \in {\cal Y}} C[x, y'] } \right)
\text{.}
\]
Then the \emph{Shannon's channel-capacity} $\C(\chan)$ of a channel $\chan$ is given by $\displaystyle \max_\pi \I(\pi, \chan)$ where $\pi$ ranges over all distributions on $\X$.
\end{definition}

Min-entropy leakage quantifies information leakage under single-attempt guessing attacks~\cite{Braun:09:MFPS,Smith:09:FOSSACS}.
\begin{definition} \label{def:MEL} \rm
Given a prior $\pi$ on $\cal X$, and a channel $\chan = ({\cal X}, {\cal Y}, C)$,
the \emph{prior vulnerability} $V(\pi)$ and the \emph{posterior vulnerability} $V(\pi, \chan)$ are defined respectively as 
%\vspace{-0.4cm}
\[
%\begin{array}{rl}
V(\pi)\!=~\displaystyle\max_{x \in {\cal X}} \pi[x]
%\\[1.5ex]
~~~\mbox{ and }~~~
V(\pi, \chan)\!=~\displaystyle\sum_{y \in {\cal Y}} \max_{x \in {\cal X}} \pi[x] C[x, y]
\text{.}
%\end{array}
\]
Then the \emph{min-entropy leakage} $\L(\pi, \chan)$ and the \emph{min-capacity} $\MC(\chan)$ are defined by:
%\vspace{-0.4cm}
\[
%\begin{array}{rl}
\L(\pi, \chan)\!=~ -\log V(\pi) + \log V(\pi, \chan)
%\\
~~~\mbox{ and }~~~
\MC(\chan)\!=~ \displaystyle\sup_{\pi} \L(\pi, \chan)
\text{.}
%\end{array}
\]
\end{definition}

\section{Information Leakage of Scheduler-Dependent Systems}
\label{sec:comp-view}

\subsection{Traces and Systems}
\label{ssec:trace}

In general the output of an information-theoretic channel can be defined in many different ways.
In this work we consider traces, or sequences of actions, as observable values.
Assume a countable set of {\em names} denoted $m,m',m_1,m_2,\ldots$ and a
countable set of {\em values} $v,v_1,v',\ldots$.
We define an {\em action} 
by $\mu,\alpha,\beta \mathbin{::=} \tau \mid \oap m v\;$.
Here $\tau$ denotes the traditional \emph{silent} or \emph{internal} action that
contains no further information.
The {\em output} action $\oap m v$ can be considered to exhibit some value $v$ via
some named mechanism $m$.
In concurrency theory the output action typically refers to the the named mechanism
as a \emph{channel name}, which is distinct from the notion of information-theoretic channel
used here.
Here the output action is used in a more general sense, in that $\oap m v$
exhibits some value $v$ such as runtime measured via mechanism $m$.
For example, $v$ could be runtime, electronic power usage or other value determined by the input,
and $m$ could be via direct communication/circuitry, indirect side effects, or
any other means.

A \emph{trace} is defined to be a sequence of actions of the form $\mu_1.\mu_2.\ldots.\mu_i$.
The notation $\alpha\in\mu_1.\mu_2.\ldots.\mu_i$ denotes that there exists a $j \in \{1, 2, \ldots, i \}$ such that $\mu_j=\alpha$.
Similarly a sequence of $i$ actions $\mu$ can be denoted $\mu^i$, and an empty sequence of actions
by $\emptyset$.
A \emph{system} is modeled as an information-theoretic channel $(\X, \Y, C)$ where $|\X|$ is finite and the set $\Y$ of observables is a finite set of traces.

\subsection{Scheduled Composition}
\label{ssec:composition}

In this section we model scheduler-dependent systems by introducing the notion of a \emph{scheduled composition} of information-theoretic channels, which interleaves outputs from different channels.

In~\cite{KawamotoCP14:qest} the \emph{parallel composition} $\chan_1 \compd \chan_2$ of two component channels $\chan_1$ and $\chan_2$ is defined as a channel that outputs ordered pairs consisting of the outputs of the two component channels.
That is, given two component channels $\chan_1 = (\X_1,\Y_1,C_1)$ and $\chan_2 = (\X_2,\Y_2,C_2)$,
the outputs of their parallel composition range over the ordered pairs $(y_1,y_2)$ for all $y_1\in\Y_1$ and $y_2\in\Y_2$.
This composition can be modeled using a scheduler that allows $\chan_1$ to perform the whole sequence $y_1$ of actions and some action $\sep \not\in \Y_1\cup\Y_2$ (for separating $y_1$ from $y_2$) before $\chan_2$ performs the actions in $y_2$.%
\footnote{Formally, we introduce $\chan_{\sep} = (\!\{\sep\}\!, \{\sep\}\!, (1))$ to consider the sequential execution of $\chan_1$, $\chan_{\sep}$ and $\chan_2$ in this order.} 
In this setting we can recognise which component channel each output of the composed channel came out of.

In this paper we consider more fine-grained schedulers that may allow $\chan_2$ to perform some actions before $\chan_1$ completes the whole sequence of actions.
To model such schedulers, we define the set of possible interleaving of two traces that  preserves the orders of occurrences of actions in the traces.

\begin{definition}[Interleaving of traces] \rm
\label{def:interleaving}
Let us consider two traces $y_1$ of the form $\alpha_1.\allowbreak \alpha_2.\ldots.\alpha_k$ and $y_2$ of the form $\beta_1.\beta_2.\ldots.\beta_l$.
The \emph{interleaving} $\Int(y_1, y_2)$ of $y_1$ and $y_2$ is the set of all traces of the form $\mu_1.\mu_2.\ldots.\mu_{k+l}$ s.t., 
for two sequences of distinct integers 
$1 \le i_1 < i_2 < \ldots < i_k \le k+l$ and
$1 \le j_1 < j_2 < \ldots < j_l \le k+l$,
we have $\mu_{i_m} = \alpha_m$ for all $m = 1, 2, \ldots, k$ and $\mu_{j_m} = \beta_m$ for all $m = 1, 2, \ldots, l$.
\end{definition}

\begin{definition} \rm
For two sets $\Y_1, \Y_2$ of observables, the \emph{interleaving $\Int(\Y_1, \Y_2)$ over $\Y_1$ and $\Y_2$} is defined by
$
\Int(\Y_1, \Y_2) =
\bigcup_{y_1\in\Y_1, y_2\in\Y_2} \Int(y_1, y_2)
\texttt{.}
$
The definition of interleaving is extended from two traces to $n$ traces as follows:
$\Int(y_1, \allowbreak y_2, \ldots, y_n) = \bigcup_{y' \in \Int(y_2, \ldots, y_n)} \Int(y_1, y')$.
For $n$ sets $\Y_1, \Y_2, \dots, \Y_n$ of observables, the interleaving $\Int(\Y_1, \Y_2, \dots, \Y_n)$ is defined analogously.
\end{definition}

\begin{figure*}[t]\label{fig:compositions}
\begin{center}
\hspace{-1.0pt}%
\subfloat[][Parallel composition]{
\begin{picture}(110, 34)
 \put( 37, 45){$C_1 \compd C_2$}
 \thicklines \thicklines
 \put( 35, 17){\framebox(40,18){$C_1$}}
 \put( 35, -9){\framebox(40,18){$C_2$}}

 \linethickness{1.4pt}
 \put(   3,  26){\vector(  1,  0){28}}
 \put(   3,  00){\vector(  1,  0){28}}
 \put(  77,  26){\vector(  1,  0){25}}
 \put(  77,  00){\vector(  1,  0){25}}
 \thinlines \thinlines
 \put(  10,  30){$X_1$}
 \put(  10,    4){$X_2$}
 \put(  85,  30){$Y_1$}
 \put(  85,    4){$Y_2$}
 \put(29,-13){\dashbox{1.0}(51,53){}}
\end{picture}
\label{fig:composition-separated}
}
~~~
\subfloat[][Observation of a scheduled composition]{
\begin{picture}(210, 34)
 \put( 45, 45){$\comps{\S}(C_1, C_2)$}
 \thicklines \thicklines
 \put( 35, 17){\framebox(40,18){$C_1$}}
 \put( 35, -9){\framebox(40,18){$C_2$}}
 \put(105,-10){\framebox(20,46){$\S$}}
 \put(155,-10){\framebox(20,46){$\Obs$}}

 \linethickness{1.4pt}
 \put(   3,  26){\vector(  1,  0){28}}
 \put(   3,  00){\vector(  1,  0){28}}
 \put(  77,  26){\vector(  1,  0){25}}
 \put(  77,  00){\vector(  1,  0){25}}
 \put( 127,  15){\vector(  1,  0){25}}
 \put( 177,  15){\vector(  1,  0){25}}
 \thinlines \thinlines
 \put(  10,  30){$X_1$}
 \put(  10,    4){$X_2$}
 \put(  85,  30){$Y_1$}
 \put(  85,    4){$Y_2$}
 \put( 136,  20){$Y$}
 \put( 185,  20){$Z$}
 \put(29,-13){\dashbox{1.0}(101,53){}}
\end{picture}
\label{fig:composition-sheduled}
}%
\vspace{-0.1cm}
\caption{Parallel composition and scheduled composition}
\label{fig:two-composition}
\hacks{\vspace{-0.5cm}}
\end{center}
\end{figure*}
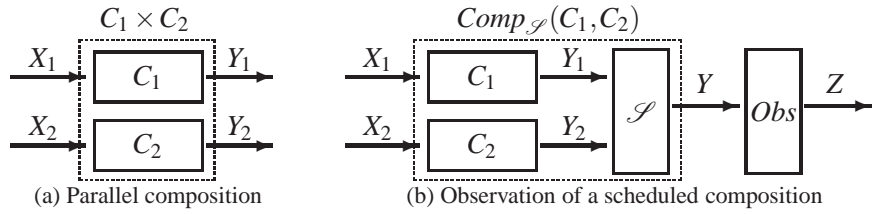

Although the interleaving defines all possible combinations of the sets of traces,
they do not define the probability of their appearance.
To reason about this,
we define a scheduler that takes two sets of traces and
probabilistically schedules their actions to form each possible trace in their interleaving.

\begin{definition}[Scheduler] \rm
\label{def:scheduler}
A {\em scheduler $\S$ on $\Y_1$ and $\Y_2$} is a function that, given two traces $y_1\in\Y_1$ and $y_2\in\Y_2$,
produces a probability distribution over all the possible interleaving $\Int(y_1, y_2)$.
We denote by $\S(y_1,y_2)[y]$ the conditional probability of having an interleaved trace $y$ given $y_1$ and $y_2$.
\end{definition}

We define a deterministic scheduler as one that produces the same output for any given two traces.
\begin{definition}[Deterministic scheduler] \rm
\label{def:det-scheduler}
A scheduler $\S$ is \emph{deterministic} if for any two traces $y_1$ and $y_2$,
there exists $y\in\Int(y_1, y_2)$ such that $\S(y_1,y_2)[y] = 1$.
\end{definition}

This provides the basis for composing channels in general, however this requires
some delicacy since the interleaving of different traces may produce the same
result.
For example, given $y_1=\tau.\oap m s$ and $y_2=\tau$ then one of the possible
traces produced is $\tau.\tau.\oap m s$. However, given $y_3=\oap m s$ and
$y_4=\tau.\tau$ then the same trace $\tau.\tau.\oap m s$ could also
be produced.

Let $p(y_1, y_2)$ be the joint probability that two component channels output two traces $y_1$ and $y_2$.
Then the probability that $\S$ produces an interleaved trace $y$ is given by:
%\[
$
p(y) = \sum_{y_1\in\Y_1, y_2\in\Y_2} p(y_1,y_2)\cdot \S(y_1,y_2)[y]
\texttt{.}
$
%\]
By~\cite{KawamotoCP14:qest} we obtain $C_1[x_1, y_1] C_2[x_2,y_2] = p(y_1,y_2|x_1,x_2)$.
Hence we can define scheduled composition of channels as follows.

\begin{definition}[Scheduled composition of channels] \rm
\label{def:composition}
The \emph{scheduled composition of} two channels $\chan_1= (\X_1,\Y_2,C_1)$ and $\chan_2 = (\X_2,\Y_2,C_2)$ \emph{with respect to} a scheduler $\S$ is
define as the channel $(\X_1\times\X_2,\Int(\Y_1,\Y_2),C)$ where the matrix element for $x_1\in\X_1$, $x_2\in\X_2$ and $y\in\Int(\Y_1,\Y_2)$ is given by:
\\[-1pt]
$
C[(x_1, x_2), y] = \sum_{y_1\in\Y_1,\, y_2\in\Y_2} C_1[x_1,y_1] C_2[x_2,y_2] \S(y_1,y_2)[y]
\texttt{.}
$
\end{definition}
We denote this scheduled composition by $\comps{\S}(\chan_1, \chan_2)$.
Note that the scheduled composition of $n$ channels can be defined by adapting the scheduler $\S$ to operate over $n$ traces in the obvious manner.

\subsection{Examples of Scheduled Composition}
\label{ssec:eg-deterministic-schedulers}

This section presents some example channels and schedulers that illustrate the
main results of this paper.
For simplicity they shall all limit their secrets to the set $\X_B=\{0,1\}$,
and their outputs to the set $\Y_m=\{\oap m 0,\tau.\oap m 0,\oap m 1,\tau.\oap m 1\}$
for a parameter $m$.

Consider the channel $\chan_1=(\X_B,\Y_{m_1},C_1)$ where
$C_1$ is given by Table~\ref{table:channel-matrix-C1}.
This channel can be considered as one that half the time simply outputs the secret
via $\oap {m_1} s$ and half the time outputs the exclusive-or $\oplus$ of the secret with $1$ as in
$\tau.\oap {m_1} {s\oplus 1}$, with the $\tau$ representing the calculation effort.
Note that this channel leaks $100\%$ of the information about the secret.
Also consider the channel $\chan_2=(\X_B,\Y_{m_2},C_2)$ where
$C_2$ is given by Table~\ref{table:channel-matrix-C2}.
This channel is similar to $\chan_1$, except that the internal action $\tau$
is observable when disclosing the secret rather than its exclusive-or.
Again this channel leaks all the secret information.

\begin{table}[t]
\begin{center}
\begin{tabular}{lr}
\hspace{-20pt}~
\begin{minipage}{0.48\hsize}
\begin{flushleft}
\begin{small}
\begin{tabular}{rr|llll}
& \multicolumn{4}{c}{~\qquad~ observable} \\[-6pt]
& & $\oap {m_1} 0$ & \hspace{-4.5pt}$\tau. \oap {m_1} 0$ & \hspace{-4.5pt}$\oap {m_1} 1$ & \hspace{-4.5pt}$\tau.\oap {m_1} 1$\\  \cline{2-6}
& 0 & 0.5 & 0 & 0 & 0.5\\[-5pt]
\raisebox{7pt}[0cm][0cm]{secret~\hspace{-11.0pt}~} \vspace{1pt}
& 1 & 0 & 0.5 & 0.5 & 0 
\end{tabular}
\vspace{-0.3cm}
\caption{Channel matrix $C_1$}
\label{table:channel-matrix-C1}
\end{small}
\end{flushleft}
\end{minipage}
\begin{minipage}{0.48\hsize}
\begin{center}
\begin{small}
\begin{tabular}{rr|llll}
& \multicolumn{4}{c}{~\qquad~ observable} \\[-6pt]
& & $\oap {m_2} 0$ & \hspace{-4.5pt}$\tau . \oap {m_2} 0$ & \hspace{-4.5pt}$\oap {m_2} 1$ & \hspace{-4.5pt}$\tau.\oap {m_2} 1$\\ \cline{2-6}
& 0 & 0 & 0.5 & 0.5 & 0\\[-5pt]
\raisebox{7pt}[0cm][0cm]{secret~\hspace{-11.0pt}~} \vspace{1pt}
& 1 & 0.5 & 0 & 0 & 0.5 
\end{tabular}
\vspace{-0.2cm}
\caption{Channel matrix $C_2$}
\label{table:channel-matrix-C2}
\end{small}
\end{center}
\end{minipage}
\end{tabular}
\end{center}
\hacks{\vspace{-0.6cm}}
\end{table}

When combining channels the r\^ole of the scheduler is very significant with
respect to the information leakage.
This section defines three types of simple schedulers for illustrating the results here.

The simplest scheduler is one that outputs the first and
second observable outputs concatenated, i.e.~given $y_1$ and $y_2$ outputs $y_1.y_2$.

\begin{definition}\rm \label{def:left-first-sc}
The \emph{(left-first) deterministic sequential scheduler} $\SDS$ is defined as follows:
$\SDS(y_1,y_2)[y]$ is $1$ if $y= y_1.y_2$ and $0$ otherwise
where $y_1\in\Y_1$, $y_2\in\Y_2$ and $y\in\Y$.
\end{definition}

\begin{example} \label{eg:SDS}
The scheduled composition $\comps {\SDS}(\chan_1,\chan_2)$ w.r.t. 
$\SDS$ has the same information leakage
as the parallel composition $\chan_1 \compd \chan_2$.
This can be shown since it follows from the definition of $\SDS$ that, 
for each $y\in\Y$, $\SDS$ uniquely identifies a pair $(y_1, y_2)$ of outputs.
For instance, let us consider the prior distribution $\pi$ on $\X_1\times\X_2$ defined by 
$(0.15, 0.20, 0.30, 0.35)$.
Then, for both of the composed channels, the mutual information is about $1.926$ and the min-entropy leakage is about $1.515$.
\end{example}

Next is the \emph{fair sequential scheduler} $\SFS$ that fairly chooses between the
first or second observable and produces that in its entirety before producing the
other.

\begin{definition}\rm \label{def:fair-seq-sc}
The \emph{fair sequential scheduler} $\SFS$ is defined by
\vspace{-0.2cm}%
\small
\begin{equation*}
\begin{array}{rcll}
\SFS(y_1,y_2)[y]
&\define&
\left\{
\begin{array}{ll}
1&\mbox{if} ~ y_1 = y_2\wedge y=y_1.y_2\\[-4pt]
0.5~~&\mbox{if} ~ y_1\neq y_2\wedge (y=y_1.y_2 \vee y=y_2.y_1)\\[-4pt]
0&\mbox{otherwise.}
\vspace{-0.1cm}
\end{array}
\right.
\end{array}
\end{equation*}
\end{definition}
\normalsize
Similar to the deterministic sequential scheduler, the information leakage can be proven to be equal
to that of the parallel composition of channels for this example.

\begin{example} \label{eg:SFS}
The scheduled composition $\comps {\SFS}(\chan_1,\chan_2)$ w.r.t. 
$\SFS$ has the same information leakage as the parallel composition $\chan_1 \compd \chan_2$.
This can be shown similarly to Example~\ref{eg:SDS}.
\end{example}

Note that the leakage preservation does {\em not} hold in general as illustrated in the following example.
\begin{example} \label{eg:less-leakage}
Consider when $\Y_1=\{\tau,\tau.\tau\}$ and $\Y_2=\{\oap m 0,\tau.\oap m 0\}$.
The observed output $\tau.\tau.\oap m 0$ can arise from
$\S(\tau,\tau.\oap m 0)$ and $\S(\tau.\tau,\oap m 0)$, where $\S$
can be $\SDS$ or $\SFS$.
Thus, both 
the schedulers $\SDS$ and $\SFS$
may allow less information leakage than the parallel composition.
\end{example}

The third example scheduler is the \emph{fair interleaving scheduler} $\SFI$ that evenly chooses the next action from the two observables.
\begin{definition}\rm \label{def:fair-inter-sc}
The \emph{fair interleaving scheduler} $\SFI$ is recursively defined as
\vspace{-0.2cm}%
\small
\begin{equation*}
\begin{array}{c}
\SFI(y_1,y_2)[y]
\define
\left\{
\begin{array}{ll}
0.5 \SFI(y_1',y_2)[y']&
\mbox{if } ~ y\!=\!\alpha.y'
\wedge\ y_1\!=\!\alpha.y_1'
\wedge\ y_2\!=\!\beta.y_2'
\wedge\ \alpha\!\neq\!\beta
\\[-3pt]
0.5 \SFI(y_1,y_2')[y']&
\mbox{if } ~ y\!=\!\beta.y'
\wedge\ y_1\!=\!\alpha.y_1'
\wedge\ y_2\!=\!\beta.y_2'
\wedge\ \alpha\!\neq\!\beta
\\[-3pt]
0.5 \SFI(y_1',y_2)[y'] + 0.5 \SFI(y_1,y_2')[y']
%~\hspace{-77pt}\\[-3pt]
&
\mbox{if } ~ y\!=\!\alpha.y'
\wedge\ y_1\!=\!\alpha.y_1'
\wedge\ y_2\!=\!\alpha.y_2'
\\[-3pt]
1&
\mbox{if } ~ (y=y_1
\wedge\ y_2=\emptyset )
\vee\ (y=y_2
\wedge\ y_1=\emptyset)
\\[-3pt]
0&
\mbox{otherwise.}
\vspace{-0.1cm}
\end{array}
\right.
\end{array}
\end{equation*}
\normalsize
\end{definition}
The fair interleaving scheduler $\SFI$ turns out to often have impact on the leakage
compared to the parallel composition of channels.
This can occur in a variety of ways and shall be explored in detail later.

\begin{example} \label{eg:SFI}
The scheduled composition $\comps {S_F}(\chan_1,\chan_2)$ w.r.t. 
$\SFI$ has less information leakage than the parallel composition $\chan_1 \compd \chan_2$.
This can be shown by considering
when the output $y$ is of the form $\tau.\oap {m_1} 0.\oap {m_2} 0$, which
can arise from both $\SFI(\tau.\oap {m_1} 0,\oap {m_2} 0)$ and
$\SFI(\oap {m_1} 0,\tau.\oap {m_2} 0)$.
Since $y$ does not uniquely identify the outputs $y_1$ and $y_2$,\, 
$\SFI$ could allow less leakage than the parallel composition.
For instance, for the prior $(0.15, 0.20, 0.30, 0.35)$, the mutual information of the scheduled
composition w.r.t. $\SFI$ is $1.695$.
This is less than those of the parallel composition and scheduled composition w.r.t. $\SDS$ in Example~\ref{eg:SDS} (both $1.926$),
thus the scheduler here alone is responsible for reducing the leakage.
\end{example}

\section{Information Leakage to Observers}
\label{sec:observed-leak}

\hacks{\vspace{-0.2cm}}
\subsection{Observers}
\label{ssec:observer}
Many kinds of capabilities of observing systems have been considered;
e.g.~an observer for strong bisimulation $\sim_s$ can recognise the internal action:  $\tau.\oap m {v} \not\sim_s \oap m {v}$, while one for weak bisimulation $\sim_w$ cannot: $\tau.\oap m {v} \sim_w \oap m {v}$.
To model different kinds of capabilities of observation, we define an \emph{observer's views} $\Z$ as the set of values recognised by the observer.
For example, $\tau.\oap m {v}$ and $\oap m {v}$ fall into two different views to an observer for strong bisimulation, but to the same view to an observer for weak bisimulation. 

We formalise the notion of an observer using a matrix that defines relationships between observable outputs of systems and the observer's views. In particular, we allow for probabilistic accuracy in observation;
that is the observer may not be perfectly accurate in identifying an output.

\begin{definition}[Generalised observer] \rm
An \emph{observer} $\O$ is defined as a triple $(\Y, \Z, \Obs)$ consisting of a finite set $\Y$ of observables, a finite set $\Z$ of observer's views and an \emph{observer matrix} $\Obs$ each of whose row represents a probability distribution; i.e., for all $y \in \Y$ we have $\sum_{z \in \Z} Obs[y, z] = 1$.
Each matrix element $\Obs[y, z]$ represents the probability that the observer has the view $z$ when the actual output is $y$.
\end{definition}

The observation matrix $\Obs$ describes the capability of the attacker to distinguish between traces.
This capability of observation 
has been formalised as an equivalence relation between states of a system in prior work~\cite{BMLW13}.
In fact, an equivalence relation $\sim$ between traces characterises a class of observers.

\begin{definition}[$\sim$-observer] \rm
Given an equivalence relation $\sim$ on $\Y$, an observer $(\Y, \Z, \allowbreak \Obs)$ is called a \emph{$\sim$-observer} 
if, for all $y_1, y_2 \in \Y$,\, $y_1 \sim  y_2$ is logically equivalent to $\Obs[y_1, z] = \Obs[y_2, z]$ for all $z\in Z$.
\end{definition}

For instance, we can consider the $\sim_s$-observer for strong bisimulation $\sim_s$ and the $\sim_w$-observer for weak bisimulation $\sim_w$.
Observe that $\sim_s$ is the identity relation on traces here.
Further, note that for every observer $\O$, there exists an equivalence relation $\sim$ between traces such that $\O$ is a $\sim$-observer.
This equivalence relation $\sim$ is defined by the following:~
$\sim \define \{ (y_1, y_2) \in \Y\times\Y \mid \mbox{ for all $z\in Z$,\, } \Obs[y_1, z] = \Obs[y_2, z] \}$.
On the other hand, the observation matrix is \emph{not} uniquely determined by the equivalence relation and therefore can express a wider range of observers' capabilities than the equivalence relation.

Among $\sim$-observers, we often consider observers that always have the same view on the same trace.

\begin{definition}[Deterministic observer] \rm
We say that an observer $(\Y, \Z, \Obs)$ is \emph{deterministic} if each probability in $\Obs$ is either $0$ or $1$; i.e., for all $y \in \Y$, there exists a unique $z \in \Z$ such that $\Obs[y, z] = 1$.
\end{definition}

For any deterministic $\sim$-observer $(\Y, \Z, \Obs)$ and
any  $y_1, y_2 \in \Y$, we have $y_1 \sim y_2$ iff, for all $z \in \Z$, we have $\Obs[y_1,z] =\Obs[y_2,z] \in \{0, 1\}$.
Then this observer always detects the equivalence class $[y]_{\sim}$ of the output $y$ from any given view $z$.
For this reason, when defining a deterministic $\sim$-observer, we typically take the set $\Z$ of views as the quotient set of $\Y$ by $\sim$, and for any $y \in \Y$ and $z \in \Z$,\, $\Obs[y, z] = 1$ iff $z = [y]_{\sim}$.
For example, consider the deterministic observers corresponding to 
$\sim_s$.
\begin{example}[Deterministic $\sim_s$-observer] \rm
A deterministic $\sim_s$-observer $(\Y, \Z, \Obs)$ satisfies the property that, for any distinct $y_1, y_2 \in \Y$, there exists a $z \in \Z$ such that either $\Obs[y_1,z] = 0$ and $ \Obs[y_2,z] = 1$ or $\Obs[y_1,z] = 1$ and $ \Obs[y_2,z] = 0$.
Therefore this observer always detects the output $y$ of the channel from any given view $z$.
For this reason we call a deterministic $\sim_s$-observer a \emph{perfect observer}.
\end{example}

Various kinds of bisimulations, or relations on observables, have been proposed
% \cite{WHHSB06sigref},
and can be represented by various deterministic observers.
Indeed, other kinds of relations can also be represented; consider an observer that cannot distinguish
which source $m_i$ a value is output upon. This can be formalised by using the equivalence relation
$\sim_{ch}$ on traces that cannot distinguishes $m_1$ from $m_2$.

The last example observer here effectively ensures no leakage by
seeing all outputs as the same:
\begin{example}[Unit observer] \rm
An observer $\O = (\Y, \Z, \Obs)$ is called a {\em unit observer} if $\Z$ is a singleton.
It has the same view regardless of the outputs of the channel, thus can detect no leakage of the channel.
\end{example}

\subsection{Observed Information Leakage}

The amount of observed information leakage depends on the capability of the observer.
To quantify this we introduce the notion of \emph{observed information leakage}.

\begin{definition}[Observed information leakage] \rm
Let $\chan = (\X, \Y, C)$ be a channel and $\O = (\Y, \Z, \Obs)$ be an observer.
For each leakage measure $L \in \{ \I, \L \}$ and any prior $\pi$ on $\X$, we define \emph{observed information leakage} 
by
%\[
$
L_\O(\pi, \chan) = L(\pi, \chan\cdot\O)
$
%\]
where $\chan\cdot\O = (\X, \Z, C\cdot\Obs)$ is the cascade composition~\cite{Espinoza:11:FAST} of $\chan$ and $\O$.
Similarly, for each $L \in \{ \C, \MC \}$, we define $L_\O(\chan) = L(\chan\cdot\O)$.
\end{definition}

We present properties of observed information leakage as follows.
The first remark is that, for each equivalence relation $\sim$ on traces, all deterministic $\sim$-observers give the same observed leakage values.
\begin{proposition}\label{prop:unique-det-sim-obs}
Let $\pi$ be any prior on $\X$ and $\chan = (\X, \Y, C)$ be any channel.
For any equivalence relation $\sim$ on $\Y$ and any two deterministic $\sim$-observers $\O_1$, $\O_2$, we have $L_{\O_1}(\pi, \chan) = L_{\O_2}(\pi, \chan)$ for $L \in \{ \I, \L \}$ and $L_{\O_1}(\chan) = L_{\O_2}(\chan)$ for $L \in \{ \C, \MC \}$.
\end{proposition}
\withproof{
\begin{proof}
The two observer matrices of $\O_1$ and $\O_2$ are identical when removing the columns with all zeros and reordering the other columns.
Therefore the observed leakage values with $\O_1$ and $\O_2$ coincide.
\end{proof}}

The following states that the deterministic $\sim_s$-observers and unit observers respectively have the maximum and minimum capabilities of distinguishing traces.
That is, the deterministic $\sim_s$-observer can detect every behaviour of the channel accurately and does not alter the leakage of the channel in any manner,
while the unit observers cannot detect any leakage of the channel.
\begin{proposition}\label{prop:obs-max-min}
For each $L \in \{ \I, \L \}$,\, $\displaystyle 0 \le L_\O(\pi, \chan) \le L(\pi, \chan)$.
For each $L \in \{ \C, \MC \}$,\, $\displaystyle 0 \le L_\O(\chan) \le L(\chan)$.
In these inequations, the left equalities hold when $\O$ is a unit observer, and the right ones hold when $\O$ is a deterministic $\sim_s$-observer.
\end{proposition}
\withproof{
\begin{proof}
If $L = \I$, then it follows from the data-processing inequality~\cite{Cover:06:BOOK} that $L(\pi, \chan\cdot\O) \le L(\pi, \chan)$.
Similar for the case $L = \C$.
If $L = \L$, then it follows from a property of the cascade composition~\cite{Espinoza:11:FAST} that $L(\pi, \chan\cdot\O) \le L(\pi, \chan)$.
Similar for the case $L = \MC$.

When $\O$ is a unit observer, the cascade $\chan\cdot\O$ is a $\#\X \times 1$-matrix.
Hence, for each $L \in \{ \I, \L \}$,\, $L(\pi, \chan\cdot\O) = 0$ regardless of $\pi$ and $\chan$.
Therefore we obtain $L(\chan\cdot\O) = 0$ for each $L \in \{ \C, \MC \}$.

When $\O$ is a deterministic $\sim_s$-observer, we obtain an identity matrix from $\Obs$ by removing the column with all zeros and by sorting the order of columns.
Therefore, for each $L \in \{ \I, \L \}$, we have $L_\O(\pi, \chan) = L(\pi, \chan)$ and, for each $L \in \{ \C, \MC \}$, we have $L_\O(\chan) = L(\chan)$.
\end{proof}}

Next we compare the capabilities of generalised observers.
Recall the composition-refinement relation $\sqsubseteq_\circ$ on channels~\cite{Alvim:12:CSF,McIverMSEM14}:
A channel $\chan_1$ is \emph{composition-refined} by another $\chan_2$, written as $\chan_1 \sqsubseteq_\circ \chan_2$, iff there exists a channel $\chan'$ such that $\chan_1 = \chan_2 \cdot \chan'$.
Since the generalised observers are also channels, we can consider this ordering $\sqsubseteq_\circ$ on observers.
For example, the unit observer is composition-refined by $\sim_w$-observers,
and the deterministic $\sim_w$-observer is by the deterministic $\sim_s$-observer.
For another example, any probabilistic $\sim_a$-observer is composition-refined by the deterministic $\sim_a$-observer:
\begin{proposition} \label{prop:ordering-deterministic}
Given any equivalence relation $\sim_a$ on $\Y$ let $\O_1 = (\Y, \Z, \Obs_1)$ and $\O_2 = (\Y, \Z, \Obs_2)$ be two $\sim_a$-observers.
If $\O_2$ is deterministic then $\O_1 \sqsubseteq_\circ \O_2$.
\end{proposition}
\withproof{
\begin{proof}
By the definition of $\sqsubseteq_\circ$, it suffices to construct a channel $\chan'$ such that  $\O_1 = \O_2\cdot\chan'$.
In fact we can construct such a $\chan'$ by choosing distinct rows of the matrix $\Obs_1$ and reordering them.
\end{proof}}

The composition-refined observer will observe less information leakage.
\begin{theorem} \label{thm:ordering-observers}
Let $\O_1$ and $\O_2$ be two observers such that $\O_1 \sqsubseteq_\circ \O_2$.
Then, for any prior $\pi$ and any channel $\chan$, we have 
$L_{\O_1}(\pi, \chan) \le L_{\O_2}(\pi, \chan)$ for $L \in \{\I, \L\}$ and
$L_{\O_1}(\chan) \le L_{\O_2}(\chan)$ for $L \in \{\C, \MC\}$.
\end{theorem}
\withproof{
\begin{proof}
For $L \in \{\I, \C\}$, we obtain the claim from the data processing inequality.
For $L \in \{\C, \MC\}$, the claim follows from a result for cascade composition~\cite{Espinoza:11:FAST}.
\end{proof}}

These results imply that no probabilistic $\sim$-observer detect more leakage than deterministic ones.

\subsection{Examples of Deterministic Observers}
\label{ssec:eg-deterministic-observation}
Theorem~\ref{thm:ordering-observers} implies that the deterministic $\sim_s$-observer does not observe less information leakage than the deterministic $\sim_w$-observer.

\begin{example} \label{eg:ordering-observers}
Let us consider the scheduled compositions in Examples~\ref{eg:SDS} and~\ref{eg:SFS} in Section~\ref{ssec:eg-deterministic-schedulers}.
Both the composed channels leak all secrets without considering observers;
i.e., they do so in the presence of $\sim_s$-observer.
On the other hand, they leak no secrets to a weakly-bisimilar observer.
For example, for each $i\in\{1,2\}$, we define the deterministic $\sim_w$-observer $\O_i$  as
$(
 \{\oap {m_i} 0,\tau.\oap {m_i} 0,\oap {m_i} 1,\tau.\oap {m_i} 1\}, \allowbreak
 \{[ \oap {m_i} 0 ]_{\sim_w}, [ \oap {m_i} 1 ]_{\sim_w}\}, \allowbreak
 \Obs)$
where $\Obs$ is the matrix given in Table~\ref{table:obs-matrix-Obs}.
Applying the $\sim_w$-observer $\O_i$ to both $\chan_1$ and $\chan_2$ yields the same matrix presented in Table~\ref{table:composed-matrix-K-Obs}.
Then both channels leak no information to the $\sim_w$-observer.
Therefore, the deterministic $\sim_s$-observer observes more information leakage than the deterministic $\sim_w$-observer also in this example.
\end{example}

\begin{table}[t]
\begin{center}
\begin{tabular}{lr}
\hspace{-14.5pt}~
\begin{minipage}{0.48\hsize}
\begin{flushleft}
\begin{small}
\begin{tabular}{rr|ccc}
& \multicolumn{4}{c}{~\qquad~\qquad~\qquad~\qquad~~~~~~ view} \\[-6.0pt]
& & $[ \oap {m_i} 0 ]_{\sim_w}$ & $[ \oap {m_i} 1 ]_{\sim_w}$ \\ \cline{2-4}
& $\oap {m_i} 0$ ~or~ $\tau.\oap {m_i} 0$ & $1$ & $0$\\[-4.5pt]
\raisebox{8pt}[0cm][0cm]{output~\hspace{-10.0pt}~} \vspace{1pt}
& $\oap {m_i} 1$ ~or~ $\tau.\oap {m_i} 1$ & $0$ & $1$
\end{tabular}
\vspace{-7.5pt}
\caption{Observer matrix $\Obs$~\hspace{-30pt}~}
\label{table:obs-matrix-Obs}
\end{small}
\end{flushleft}
\end{minipage}
~~~
\begin{minipage}{0.48\hsize}
\begin{center}
\begin{small}
\begin{tabular}{rr|cc}
& \multicolumn{3}{c}{~ view} \\[-5.0pt]
& & $[ \oap {m_i} 0 ]_{\sim_w}$ & $[ \oap {m_i} 1 ]_{\sim_w}$\\ \cline{2-4}
& 0 & 0.5 & 0.5 \\[-4.5pt]
\raisebox{8pt}[0cm][0cm]{secret~\hspace{-10.0pt}~} \vspace{1pt}
& 1 & 0.5 & 0.5
\end{tabular}
\vspace{-7.5pt}
~\qquad~
\caption{Composed matrix $C_i\cdot\Obs$}
\label{table:composed-matrix-K-Obs}
\end{small}
\end{center}
\end{minipage}
\end{tabular}
\end{center}
\hacks{\vspace{-0.2cm}}
\end{table}

The scheduled composition can also leak more information than the parallel composition (and even than each component channel) in the presence of imperfect observers.

\begin{example}[Observer dependent]
\label{ex:observer-depend}
Consider the scheduled composition of the channels $\chan_1$ and $\chan_2$ w.r.t. the fair interleaving scheduler $\SFI$.
By Example~\ref{eg:SFI}, the leakage of the scheduled composition w.r.t. $\SFI$ is less than that of the parallel composition in the presence of the deterministic $\sim_s$-observer.

However, the leakage of the scheduled composition is more than that of the parallel composition (and even than that of each component channel) when the $\sim_w$-observer $\O$ is being considered; 
e.g., $\L_{\O}( \pi, \comps{\SFI}(\chan_1, \chan_2) ) = 0.215 > 0 = \L_{\O}( \pi, \chan_1 \compd \chan_2 ) = \L_{\O}( \pi, \chan_1)$ for $\pi = (0.15, 0.20, 0.30, 0.35)$.
\end{example}

\subsection{Example of Probabilistic Observers}
\label{ssec:prob-observation}
The notion of deterministic $\sim$-observers is useful to model various observers, but they may not cover all realistic settings.
For example, when the internal action $\tau$ represents time to perform internal computation, observers may recognise it only probabilistically, for instance with probability $0.7$.
Then such \emph{probabilistic observers} cannot be modeled as deterministic observers but as generalised observers, which quantify the capabilities of probabilistic observation.
As far as we know, no previous work on quantitative information flow analyses have considered probabilistic observers.

\begin{example}
\label{eg:probabilistic-observer}
Consider a probabilistic observer $\O$ that can recognise a single internal action $\tau$ only probabilistically but two or more consecutive $\tau$'s with probability $1$.
For instance, $\O$ recognises the trace $(\tau.\oap {m_i} 0. \oap {m_i} 1)$ correctly with probability $0.7$ and confuses it with either $(\oap {m_i} 0. \oap {m_i} 1)$,\, $(\oap {m_i} 0. \tau. \oap {m_i} 1)$ or $(\oap {m_i} 0. \oap {m_i} 1. \tau)$ each with probability $0.1$.
Consider the schedule-composed channel $\comps {\SFI} (\chan_1, \chan_2)$ from Example~\ref{eg:SFI}.
The observed mutual information is $0.783$ under the probabilistic observer $\O$, 
which is between $0.090$ and $1.695$ as observed under the deterministic $\sim_w$-observer and $\sim_s$-observer.
\end{example}

\section{Relationships between Scheduling and Observation}
\label{sec:main}

This section generalises the previous examples to show three kinds of results.
First, we identify conditions on component channels under which leakage cannot be effected by the scheduled composition.
Second, we show that scheduled composition can leak more or less information than the
parallel composition, including results on the bounds of the information leaked.
Third, we present an algorithm for finding a scheduler that minimises the min-entropy  leakage/min-capacity under any observer

\subsection{Information Leakage Independent of Scheduling}
\label{ssec:perfect}

This section presents results for determining when the leakage is independent of the scheduler.
Regardless of the scheduler and observer, the leakage of the scheduled composition is equivalent to that of the parallel composition under certain conditions on component channels that are detailed below.

\begin{theorem}
\label{thm:no-shared-interleavings}
Let $\chan_1=(\X_1,\Y_1,C_1)$ and $\chan_2=(\X_2,\Y_2,C_2)$ be channels.
Assume that, for any $y_1,y'_1\in\Y_1$ and $y_2,y'_2\in\Y_2$,\,
if $\Int(y_1,y_2)\cap \Int(y'_1,y'_2)\neq\emptyset$ then $y_1 = y'_1$ and $y_2 = y'_2$.
Then, for every scheduler $\S$ and observer $\O$,
the leakage of the scheduled composition is the same as that of the parallel composition.
\end{theorem}
\withproof{
\begin{proof}
Observe that for each $y\in\Int(\Y_1,\Y_2)$ there is a unique $(y_1, y_2) \in \Y_1\times\Y_2$ that yields the interleaved trace $y$.
That is, given any output $y$ of the scheduled composition, the observer uniquely identifies the component traces $(y_1, y_2)$.
Thus the leakage of the scheduled composition is equivalent to that of the parallel composition.
\hspace*{\fill} $\Box$
\end{proof}}

By adding a stronger requirement to Theorem~\ref{thm:no-shared-interleavings}, we obtain the following corollary.
\begin{corollary}
\label{cor:no-shared-actions}
Let $\chan_1=(\X_1,\Y_1,C_1)$ and $\chan_2=(\X_2,\Y_2,C_2)$ be channels.
Assume that, for all $(y_1, y_2)\in\Y_1\times\Y_2$, $\alpha\in y_1$ and $\beta\in y_2$, we have $\alpha\neq \beta$.
Then, for every scheduler $\S$ and observer $\O$,
the leakage of the scheduled composition is the same as that of the parallel composition.
\end{corollary}
\withproof{
\begin{proof}
Since there are no actions shared between the traces of $\Y_1$ and $\Y_2$, 
we have that, for any $y_1,y'_1\in\Y_1$ and $y_2,y'_2\in\Y_2$,
if $\Int(y_1,y_2)\cap \Int(y'_1,y'_2)\neq\emptyset$
then $y_1=y'_1$ and $y_2=y'_2$.
By Theorem~\ref{thm:no-shared-interleavings} the claim follows.
\hspace*{\fill} $\Box$
\end{proof}}

\subsection{Schedulers for Altering Information Leakage}
\label{ssec:alter-leakage}

This section considers when schedulers can alter the leakage of a scheduled composition.
This is distinct from prior results where it has been shown that the composition cannot leak more information than the component channels~\cite{barthe2011information,espinoza2013min,KawamotoCP14:qest},
since here more information can be leaked to imperfect observers.

In general scheduled composition can yield more or less leakage than the individual component channels or
their parallel composition. This is illustrated by
Example~\ref{ex:observer-depend}. Unfortunately heuristics for determining 
when more information is leaked end up being rather complicated and dependent on  many
relations between traces, interleavings, equivalences, and then subject to generalities about
both schedulers and observers.
Ultimately it is easier to show by examples that, for some channels, prior, and $\sim$-observer, 
there is a scheduler by which the scheduled composition leaks strictly more information than the parallel composition.
Since this clearly holds by example, 
we consider a class of schedulers
under which the scheduled composition does not leak more information than the parallel composition.

To define this we extend an equivalence relation $\sim$ on traces to probability distributions of traces:
We say that two distributions $D$ and $D'$ on a set $\Y$ are \emph{$\sim$-indistinguishable} (written as $D \sim D'$) if the deterministic $\sim$-observer cannot distinguish $D$ from $D'$ at all, i.e., 
for all equivalence class $t \in \Y/\!\sim$,\, we have $\sum_{y \in t} D[y] = \sum_{y \in t} D'[y]$.
Using $\sim$-indistinguishability we define a scheduler that does not leak any behaviour of the system that the $\sim$-observer cannot detect.

\begin{definition}\rm \label{def:sim-blind-sc}
Let $\sim$ be an equivalence relation on $\Y_1\cup\Y_2\cup\Int(\Y_1,\Y_2)$.
A scheduler $\S$ on $\Y_1$ and $\Y_2$ is a \emph{$\sim$-blind scheduler} when,
for any two pairs $(y_1,y_2), (y'_1,y'_2) \in \Y_1\times\Y_2$,\,
we have $y_1 \sim y'_1$ and $y_2 \sim y'_2$ iff we have $\S(y_1,y_2) \sim \S(y'_1,y'_2)$.
\end{definition}
For instance, the deterministic sequential scheduler $\SDS$ and the fair sequential scheduler $\SFS$ are $\sim_w$-blind while the fair interleaving scheduler $\SFI$ is not.
Note that $\sim$-blind schedulers 
do not leak any behaviour that would not be visible to the deterministic $\sim$-observers.
Thus they do not gain more information from the scheduled composition w.r.t.
$\sim$ than the parallel composition.

\begin{theorem} \label{thm:sim-blind-sc}
Let $\pi$ be a prior, $\chan_1$ and $\chan_2$ be two channels, $\O$ be a deterministic  $\sim$-observer, and $S$ be a $\sim$-blind scheduler.
For each $L \in \{ \I, \L \}$ we have $\displaystyle L_\O(\pi, \allowbreak\comps \S(\chan_1,\allowbreak\chan_2)) \le L_\O(\pi, \chan_1\compd\chan_2)$.
For each $L \in \{ \C, \MC \}$ we have $\displaystyle L_\O(\comps{\S}(\chan_1,\allowbreak\chan_2)) \le L_\O(\chan_1\compd\chan_2)$.
When $\S$ is also deterministic, the leakage relations become equalities.
\end{theorem}
\withproof{
\begin{proof}
By Definition~\ref{def:sim-blind-sc}, the cascade $\S\cdot \O$ is also a $\sim$-observer.
Since $\O$ is the deterministic $\sim$-observer, $\S\cdot \O \sqsubseteq_\circ \O$ follows from Proposition~\ref{prop:ordering-deterministic}.
Hence the leakage of the scheduled composition w.r.t. $\sim$ is upper-bounded by that of the parallel composition.

When $\S$ is also deterministic, then the cascade $\S \O$ is a deterministic observer.
By Proposition~\ref{prop:unique-det-sim-obs} the leakages under the observers $\S \O$ and $\O$ are equivalent.
\hspace*{\fill} $\Box$
\end{proof}
}

For instance, since $\SDS$ and $\SFS$ are $\sim_w$-blind schedulers, the deterministic $\sim_w$-observers do not gain more information from the scheduled composition w.r.t. $\sim_w$ than the parallel composition.
In fact, they have the same leakage in Example~\ref{eg:ordering-observers}.

The following result is about a heuristic for when leakage can be changed by the
properties of the scheduler.
This is presented here to clarify the properties.

\begin{theorem} \label{thm:shed-alter}
Let $\chan_1=(\X_1,\Y_1,C_1)$ and $\chan_2=(\X_2,\Y_2,C_2)$ be two channels.
Assume that there exist $y_1,y'_1\in \Y_1$ and $y_2,y'_2\in\Y_2$ such that
$\Int(y_1,y_2)\cap\Int(y'_1,y'_2)\neq\emptyset$.
Then it is possible for the scheduled-composition of $\chan_1$ and $\chan_2$
to alter the mutual information and min-entropy leakage for some prior.
\end{theorem}
\withproof{
\begin{proof}
By assumption, we have  either $y_1\neq y'_1$ or $y_2\neq y'_2$.
Consider when $y_1\neq y'_1$ and $y_2=y'_2$. We have that there exists
$y\in\Int(y_1,y_2)\cap\Int(y'_1,y_2)$ and so any scheduler $\S$ where
$\S(y_1,y_2)[y] > 0$ and $\S(y'_1,y_2)[y]>0$ will yield a matrix where
observable $y$ does not uniquely identify $y_1$ or $y_1'$, yet each row of the matrix
for $\Y_1\times \Y_2$ does uniquely identify each $y_1\in\Y_1$.
Then conclude by observing that this creates an inequality on the information
leakage of $\comps \S (\chan_1,\chan_2)$ compared to $(\chan_1\compd\chan_2)$.
\hspace*{\fill} $\Box$
\end{proof}
}

\subsection{Schedulers for Minimising Information Leakage}
\label{ssec:algorithms}

This section presents results for finding a scheduler that minimises the min-entropy leakage and min-capacity in the presence of any observer.

\begin{theorem}
\label{thm:sc-for-minimise-MEL}
Given any prior $\pi$, two channels $\chan_1$, $\chan_2$ and any observer $\O$,
there is an algorithm that computes a scheduler $\S$ that minimises the observed min-entropy leakage $\L_\O(\pi, \comps{\S}(\chan_1, \chan_2))$ of the scheduled composition.
\end{theorem}
\begin{proof}
To find a scheduler $\S$ that minimises the observed min-entropy leakage $\L_\O(\pi, \allowbreak \comps{\S}(\chan_1, \chan_2))$,
it is sufficient to find $\S$ that minimises the observed posterior vulnerability $V(\pi, \comps{\S}(\chan_1, \chan_2)\cdot \O)$.

For $(x_1, x_2) \in \X_1 \times \X_2$ and $(y_1, y_2) \in \Y_1 \times \Y_2$,
let $p(x_1, x_2, y_1, y_2) = \pi[x_1, x_2] (C_1\times C_2)[(x_1.x_2), (y_1,y_2)]$.
For each $z \in \Z$ let
$
%\[
v_z = 
\hspace{-0.1cm}
\max_{(x_1, x_2)\in \X_1\times\X_2}
%\hspace{-0.35cm}
\sum_{
y_1, y_2, y
}\hspace{-0.0cm}
p(x_1, x_2, y_1, y_2) \S(y_1,y_2)[y] \Obs[y,z]
%
%\]
$
where $(y_1, y_2)$ and $y$ range over $\Y_1\times\Y_2$ and $\Int(\Y_1,\Y_2)$ respectively.
Let $Pos(y_1, y_2)[y]$ be the $(|\Y_1|\times|\Y_2|, |\Int(\Y_1,\Y_2)|)$-matrix defined by the following:
$Pos(y_1, y_2)[y] = 1$ if $y$ can be obtained by interleaving $y_1$ and $y_2$, and
$Pos(y_1, y_2)[y] = 0$ otherwise.

To find a scheduler matrix $\S$ that minimises the observed posterior vulnerability, it suffices to solve the linear program that minimises $\sum_{z \in \Z} v_z$,
subject to
% DO NOT REMOVE OR CHANGE THE STYLE & COMMENTS BELOW
\begin{itemize}
\item for each $(x_1, x_2, z)\!\in\!\X_1\!\times\!\X_2\!\times\!\Z$,
%\item for each $(x_1, x_2)\in \X_1\times\X_2$ and $z \in \Z$,
$
%\displaystyle
\sum_{
y_1, y_2, y
%\hbox{$
%\begin{scriptsize}
%\begin{array}{l}
%y \in \Int(\Y_1, \Y_2) \\[-0.1cm]
%y_1\in\Y_1, y_2\in\Y_2
%\end{array}
%\end{scriptsize}$}
}\,
p(x_1, x_2, y_1, y_2) \S\!(y_1,y_2)[y] \Obs[y,z]\!\le~v_z
$
\vspace{-0.1cm}
\item for each $(y_1, y_2) \in \Y_1 \times \Y_2$,\,
$
%\displaystyle
\sum_{y
%y \in \Int(\Y_1, \Y_2)
}\,
Pos(y_1, y_2)[y] \S(y_1,y_2)[y]
= 1.
$
\end{itemize}
Note that the second constraint means that each row of the scheduler matrix $\S$ must sum to $1$.
In this linear program, the scheduler matrix element $\S(y_1, y_2)[y]$ for each $(y_1,y_2) \in \Y_1\times\Y_2$ and $y\in \Int(\Y_1,\Y_2)$ and $v_z$ for each $z \in \Z$ are variables.
We can solve this problem using the simplex method or interior point method.
(In practice, we can efficiently solve it using a linear programming solver such as \lpsolve{}~\cite{lpsolve}.)
Hence we obtain a scheduler matrix $\S$ that minimises $\sum_{z \in \Z} v_z$.
\end{proof}

In the above linear program the number of variables is $|\Y_1|\times|\Y_2|\times|\Int(\Y_1,\Y_2)| + |\Z|$, and the number of constraints is $|\X_1|\times|\X_2|\times|\Z| + |\Y_1|\times|\Y_2|$.
Since the number of interleaved traces grows exponentially in the number of traces, the time to compute a minimising scheduler is exponential in the number of component traces.
When the observer $\O$ is imperfect enough for $|\Z|$ to be very small, then the computation time improves significantly in practice.
On the other hand, when the number of traces is very large, we may heuristically obtain a scheduler with less leakage by results in the previous section.

To obtain a scheduler that minimises the worst-case leakage value, it suffices to consider a scheduler that minimises the min-capacity.

\begin{corollary}
\label{cor:sc-for-minimise-MC}
Given two channels $\chan_1$, $\chan_2$ and any observer $\O$,
there is an algorithm that computes a scheduler $\S$ that minimises the observed min-capacity 
of the scheduled composition.
\end{corollary}
\withproof{
\begin{proof}
The min-capacity is obtained when the prior $\pi$ is uniform.
Therefore we obtain a minimizing scheduler $\S$ by the algorithm in Theorem~\ref{thm:sc-for-minimise-MEL} using the uniform prior $\pi$.
\end{proof}}

These two results give the minimum amount of leakage that is possible for any scheduling.

\begin{example}\label{eg:minimise-leak}
Consider the channels $\chan_1, \chan_2$ defined in Section~\ref{ssec:eg-deterministic-schedulers}.
By Theorem~\ref{thm:sc-for-minimise-MEL}, the minimum observed min-entropy leakage w.r.t. the prior $(0.15, 0.20, 0.30, \allowbreak 0.35)$ is $1.237$ under the deterministic $\sim_s$-observer, and $0.801$ under the probabilistic observer defined in Example~\ref{eg:probabilistic-observer}.
By Corollary~\ref{cor:sc-for-minimise-MC}, the minimum observed min-capacity is $1.585$ under the deterministic $\sim_s$-observer, and $1.138$ under the probabilistic observer.
\end{example}

Since the channel capacity will not exceed the min-capacity~\cite{smith:qest11}, 
the minimum observed min-capacity obtained by the above algorithm gives an upper bound on the minimum channel capacity.

\hacks{\vspace{-0.4cm}}
\newcommand{\votezero}{\oap m 0}
\newcommand{\voteone}{\oap m 1}

\section{Case Studies}
\label{sec:eval}

\hacks{\vspace{-0.2cm}}
\subsection{Sender Anonymity}
\label{ssec:vote}

In e-voting \emph{sender anonymity} can be summarised as the issue of collecting votes from a number of voters
and being able to expose the aggregate vote information while revealing as little as possible
about how each voter voted.
This can be solved by a general application of a mix network~\cite{Chaum:81}
where all the votes are sent via mixing systems that output the votes in a manner that
should not reveal how each voter voted.

This can be represented here by each voter being an information-theoretic channel that
outputs their vote.
For example, consider a simple voting in which possible votes are $0$ and $1$ and each voter outputs the chosen vote via $\votezero$ or $\voteone$, respectively.
Then each voter (indexed by $i$) can be represented by the channel
$\chan_i=(\{0,1\},\{\votezero,\voteone\},\allowbreak C_i)$ where
$C_i[k,\oap m k]=1$ for $k\in\{0,1\}$ and each voter has a prior $\pi_i$ on $\{0,1\}$.
Observe that each such voter channel alone fully reveals the prior for the channel.

The scheduled composition of the voters represents the mix network with the schedulers
representing the mixing algorithm and thus providing the ability to reason over their
effect on information leakage.
Consider the following problem with five voters $\chan_1$ to $\chan_5$.
As illustrated in Figure~\ref{fig:voters},
the ballot of each voter is sent via intermediate severs (schedulers) $\chan_{A}$, $\chan_{B}$, $\chan_{S1}$ that mix the order of ballots.
The final system $\chan_{S2}$ combines $\chan_{S1}$ and $\chan_B$ to output all the votes
according to some mixing.

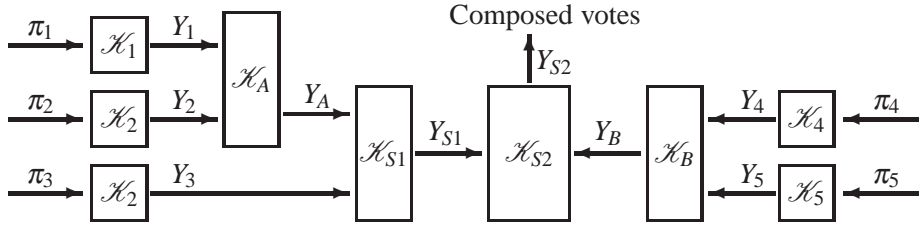
\begin{figure*}[t]
\begin{center}
\begin{picture}(350, 55)
 \thicklines \thicklines
 % voters 1-3
 \put( 35, 46){\framebox(20,20){$\chan_1$}}
 \put( 35, 18){\framebox(20,20){$\chan_2$}}
 \put( 35,-10){\framebox(20,20){$\chan_2$}}
 % voters 4-5
 \put(295, 18){\framebox(20,20){$\chan_4$}}
 \put(295,-10){\framebox(20,20){$\chan_5$}}
 % chan A
 \put(85, 18){\framebox(20,50){$\chan_A$}}
 % chan S1
 \put(135,-10){\framebox(20,50){$\chan_{S1}$}}
 % chan B
 \put(245,-10){\framebox(20,50){$\chan_{B}$}}
 % chan S2
 \put(185,-10){\framebox(30,50){$\chan_{S2}$}}

 \linethickness{1.4pt}
 % voters 1-3
 \put(  3, 56){\vector(  1,  0){28}}
 \put(  3, 28){\vector(  1,  0){28}}
 \put(  3, 00){\vector(  1,  0){28}}
 \put( 57, 56){\vector(  1,  0){25}}
 \put( 57, 28){\vector(  1,  0){25}}
 \put( 57, 00){\vector(  1,  0){75}}
 \put(107, 30){\vector(  1,  0){25}}
 % voters 4-5
 \put(347, 28){\vector( -1,  0){28}}
 \put(347, 00){\vector( -1,  0){28}}
 \put(293, 28){\vector( -1,  0){25}}
 \put(293, 00){\vector( -1,  0){25}} 
%
 % final output
 \put(200, 42){\vector( 0,  1){18}} 
 \put(203, 47){$Y_{S2}$}
 \put(170, 64){Composed votes}
 \put(157, 15){\vector(  1,  0){25}}
 \put(242, 15){\vector( -1,  0){25}}
 \thinlines \thinlines
%
 % voters 1-3
 \put( 10, 60){$\pi_1$}
 \put( 10, 32){$\pi_2$}
 \put( 10,  4){$\pi_3$}
 % voters 4-5
 \put(330, 32){$\pi_4$}
 \put(330,  4){$\pi_5$}
 % chan A
 \put( 65, 60){$Y_1$}
 \put( 65, 32){$Y_2$}
 % chan S1
 \put( 65,  4){$Y_3$}
 \put(115, 34){$Y_A$}
 % chan B
 \put(280, 32){$Y_4$}
 \put(280,  4){$Y_5$}
 % chan S2
 \put(162, 20){$Y_{S1}$}
 \put(225, 20){$Y_{B}$}
\end{picture}
\caption{Structure of composed channels for voters}
\label{fig:voters}
\hacks{\vspace{-0.5cm}}~
\end{center}
\end{figure*}
Using the deterministic sequential scheduler $\SDS$ for all compositions reveals all information on how each voter voted.
That is, the leakage is considered to be 5-bits (as each vote is 0 or 1).
On the other hand,
using the fair sequential scheduler $\SFS$ for all compositions leaks less information than $\SDS$.
When $\pi$ is uniform and $\chan$ is the composed channel in Figure~\ref{fig:voters} with the appropriate scheduling, we obtain 
$\L(\pi,\chan) = 3.426$ and $\I(\pi,\chan) = 2.836$.
Observe that
here
the third voter's output can only appear in the
1st, 3rd, or 5th position in the final trace. 
This is repaired by using the fair interleaving scheduler $\SFI$ for all compositions
that leaks even less information:
$\L(\pi,\chan) = 2.901$ and $\I(\pi,\chan) = 2.251$.

A more interesting case is when different compositions use different schedulers.
Since the votes do not contain any information about the system they came from, let alone voter.
Using the fair sequential scheduler for $\chan_A$ and $\chan_B$,
and the fair interleaving scheduler for $\chan_{S2}$, along with a
specially constructed scheduler for $\chan_{S1}$ can reduce the information leakage to a minimum.
Then the min-entropy leakage is $2.824$ and the mutual information is $2.234$.
Note that when there is only one scheduler that receives all 5 ballots, the minimum min-capacity of the composed system (over all possible schedulers) is $2.585$.

The example can be extended further by adding $\tau$ steps before votes to indicate time taken
for some parts of the process.
For a simple example, consider when voters 1 and 2 have a $\tau$ step before their vote to represent
the time taken, e.g.~as indicative of voting order, or the time taken for the extra mixing step.
In the presence of all fair interleaving schedulers, the observed min-entropy leakage and the mutual information are respectively $3.441$ and $2.785$ under the perfect observer.
However, these shift to $3.381$ and $2.597$, respectively, under the deterministic $\sim_w$-observer.

\subsection{Side-Channel Attacks}
\label{ssec:side}

Consider the small program shown in Figure~\ref{fig:program},
where an observable action is repeated in a loop.
This program captures, for instance, some aspects of decryption algorithms of certain cryptographic schemes, such as RSA.
Intuitively, $\tt X[~]$ is the binary array representing a 3-bit secret (e.g. {\tt 011}), which corresponds to secret decryption keys.
The timing of the algorithm's operation reveals which bit of the secret key is $1$,
since the observable-operation $\oap m 1$ can be detected, perhaps as power consumption, response time, or some other side-effect of the algorithm~\cite{kocher1996timing}.
We denote by $\chan$ the channel defined by this program.

Consider composition of $\chan$ with itself, e.g., when applying the algorithm
to different parts of the message in parallel.
Clearly if the parallel composition is taken then both instances of $\chan$ will
leak all their information about the key.
On the other hand, the scheduled composition may have less leakage.

We first consider the case each instance of the component channel $\chan$ receives a different secret bit string independently drawn from the uniform prior.
This captures the situation in which each decryption operation uses different secret keys.
When the fair interleaving scheduler mixes the two traces, the min-entropy leakage and the mutual information are respectively $4.257$ and $3.547$ in the presence of the perfect observer, and $2.807$ and $2.333$ in the presence of the deterministic $\sim_w$-observer.

Next we consider the case where both instances of $\chan$ share the same secret key which has been drawn from the uniform prior.
When the fair interleaving scheduler mixes the two traces, the min-entropy leakage and the mutual information are respectively $3.000$ and $3.000$ (all 3 bits of the secret key are leaked)
under the perfect observer,
and $2.000$ and $1.811$ under the deterministic $\sim_w$-observer.

\begin{figure}[t]
\begin{center}
\begin{tabular}{lr}
\begin{minipage}{0.48\hsize}
\begin{flushleft}
\begin{footnotesize}
\begin{verbatim}
for(i = 0; i < 3; i++) {
   tau;
   if(X[i] = 1) {
     m<1>; //observable-operation
   }
}
\end{verbatim}
\vspace{-14.0pt}
\caption{Decryption algorithm~\hspace{-10pt}~}
\label{fig:program}
\end{footnotesize}
\end{flushleft}
\end{minipage}
\begin{minipage}{0.48\hsize}
\begin{center}
\begin{small}
\begin{tabular}{rr|ccc}
& \multicolumn{3}{c}{\qquad\qquad\quad view} \\[-4.5pt]
& & $\tau$ & $\oap m 1$ & $\emptyset$\\[-2.0pt] \cline{2-5}
& $\tau$ & 0.8 & 0.1 & 0.1\\[-4.5pt]
\raisebox{8pt}[0cm][0cm]{output~\hspace{-10.0pt}~} \vspace{1pt}
& $\oap m 1$ & 0.05 & 0.9 & 0.05
\end{tabular}
\vspace{-12.5pt}
~\qquad\qquad~\qquad~~~\caption{Probabilistic observer matrix}
\label{fig:prob-obs-mat}
\end{small}
\end{center}
\end{minipage}
\end{tabular}
\end{center}
\hacks{\vspace{-0.6cm}}
\end{figure}

More interesting is to consider the case where the observer is only able
to detect approximate information through the side-channel.
Consider the observer $\O$ that only probabilistically observes actions
according to the matrix in Figure~\ref{fig:prob-obs-mat}.
Here $\emptyset$ indicates that nothing is detected by the attacker not even a $\tau$.
For example, applying this observer to the trace $\tau.\tau.\tau$ may yield
$\tau.\tau$ when one $\tau$ is not observed (represented $\emptyset$ in the matrix).
Such an observer is less effective even when applied to the parallel composition
of channels. However, this applies even further when applied to any scheduled
composition since the loss of information through poor detection cannot even be
limited to one channel or the other. Thus, a trace of length 5, even from a
leaky scheduler such as the (left-first) sequential scheduler, would leak less
than the parallel composition (since it would be clear which composite channel
had been poorly observed).

For instance, let us consider the case each instance of $\chan$ independently receives a secret from the uniform prior and the fair interleaving scheduler is used.
Then the min-entropy leakage and the mutual information are respectively $3.306$ and $1.454$ under this probabilistic observer.
If we consider the case both instances of $\chan$ shares the same secret, then
the leakage values are respectively $2.556$ and $1.924$.

\hacks{\vspace{-0.4cm}}
\section{Related Work}
\label{sec:related}

%\paragraph{Schedulers.}
%
Regarding schedulers there are a variety of studies on relationships between schedulers and information leakage
\cite{Chatzikokolakis07makingrandom,andres:2011:hal-00573447:1}.
In \cite{cgUCL-PLLKCSC08} 
the authors consider a \emph{task-scheduler}
that is similar to our schedulers, albeit restricted to the form of our deterministic scheduler.
The schedulers in this paper are also similar to the \emph{admissible schedulers} of \cite{andres:2011:hal-00573447:1}.
Both are defined to depend only upon the observable outputs, that is the traces they schedule.
This avoids the possibility of leakage via the scheduler being aware of the intended secret directly and so leaking information.
Differently to admissible schedulers, here the scheduler can be probabilistic, which is similar in concept to the probabilistically defined (deterministic)
schedulers of \cite{Zhang:2010:MCI:2175554.2175561}, although they explore scheduling and determinism of Markov Chains and not information leakage.

Most work on schedulers has focused on preventing any leakage at all, indeed the problem is
typically defined to prevent any high/secret information leaking. This in turn sets extremely
high requirements upon the scheduler, and so proves to be difficult to achieve, or even 
impossible.
Here we take an approach to scheduling that allows for probabilistic schedulers and so reasoning
about the quantitative information leakage, rather than total leakage.
Thus we permit schedulers that can be daemonic or angelic, as well as many in between that
may closer resemble the behaviour of real world systems.

%\paragraph{Observers.}
%
Regarding observers there is little prior work in quantitative information flow and
quantifying the capability of the observer.
\cite{BMLW13} has some similarity where they formalise an equivalence of 
system states
similar in style to the deterministic $\sim$-observers here.
However, this does not model observers as part of information-theoretic channels, hence does not allow the probabilistic behaviour of observers.

\hacks{\vspace{-0.4cm}}
\section{Conclusions and Future Work}
\label{sec:conclude}

We have introduced the notion of the scheduled composition of channels and  generalised the capabilities of the observers
to reason about 
more systems.
Then we have presented theories that can be used as heuristics to detect when scheduled composition may have
an effect on the information leakage. This determines when scheduled composition
is a potential risk/benefit to a scheduler-dependent system.
Scheduling can both leak more information, or less information
to an observer depending on many factors,
while some leakage bounds can be obtained for schedule-composed channels.
Further, we have shown an algorithm for finding a scheduler
that minimises the leakage of the scheduled composition.

The work here provides a foundation for continuing research into concurrent behavior, including interactive systems.
Here we have limited the systems to finite sets of secrets and observables since this 
aligns with the discrete version of leakage calculations. By shifting to continuous domains we
can investigate some systems with infinite secrets or observables.
Similarly the schedulers here assume finite traces and are typically defined 
over the entire possible traces. However, many do not require this, and can be defined
only upon the next action in the trace. This allows for alternate definitions without changing
the results, and easier applicability to infinite settings.

\hacks{\vspace{-0.4cm}}
\bibliographystyle{eptcs}
\bibliography{bib,short}

\end{document}